%% file: asynRefArxiv3.tex
\documentclass[aps,pra,twocolumn,superscriptaddress,groupedaddress,showpacs]{revtex4} 

\usepackage{natbib}
\usepackage{amssymb,amsmath,latexsym}
\usepackage{amsthm}
\usepackage{hyperref}
\usepackage{bm}
\usepackage{microtype}
\usepackage{url}

\usepackage[lined,ruled,vlined]{algorithm2e}

\usepackage{xcolor}
\theoremstyle{plain}

\newtheorem{lemma}{Lemma}
\newtheorem{theorem}{Theorem}

\theoremstyle{definition}
\newtheorem{definition}{Definition}

\theoremstyle{remark}

\newcommand{\abs}[1]{\left| #1 \right|}

\newcommand{\pname}[1]{\textsf{#1}}

\newcommand{\qsucc}{q_\mathrm{succ}}

\newcommand{\readyone}{ready$_1$}
\newcommand{\readytwo}{ready$_2 $}

\begin{document}

\title{Asynchronous reference frame agreement in a quantum network}
\author{Tanvirul Islam} \email[]{tanvir@locc.la}
\affiliation{School of Computing, National University of Singapore, 13 Computing Drive, 117417 Singapore}
\affiliation{Centre for Quantum Technologies, National University of Singapore, 3 Science Drive 2, 117543 Singapore}
\affiliation{Qutech, Delft University of Technology, Lorentzweg 1, 2628 CJ Delft}

\author{Stephanie Wehner} \email[]{steph@locc.la}
\affiliation{Centre for Quantum Technologies, National University of Singapore, 3 Science Drive 2, 117543 Singapore}
\affiliation{Qutech, Delft University of Technology, Lorentzweg 1, 2628 CJ Delft}

\begin{abstract} 
An efficient implementation of many multiparty protocols for quantum networks requires that all the nodes in the network share a common reference frame. Establishing such a reference frame from scratch is especially challenging in an asynchronous network where network links might have arbitrary delays and the nodes do not share synchronised clocks. In this work, we study the problem of establishing a common reference frame in an asynchronous network of $n$ nodes of which at most $t$ are affected by arbitrary unknown error, and the identities of the faulty nodes are not known. We present a protocol that allows all the correctly functioning nodes to agree on a common reference frame as long as the network graph is complete and not more than $t<n/4$ nodes are faulty. As the protocol is asynchronous, it can be used with some assumptions to synchronise clocks over a network. Also, the protocol has the appealing property that it allows any existing two-node asynchronous protocol for reference frame agreement to be lifted to a robust protocol for an asynchronous quantum network. 
\end{abstract}

\pacs{03.67.-a, 03.67.Hk, 03.67.Dd}

\maketitle

\section{Introduction}

To use quantum cryptography on a global scale one must first have a functioning quantum internet~\cite{kim08}. Recently this necessity has inspired a lot of effort in the research and development of satellite~\cite{AJP+03,BTD+09,PYB+05,BAM+06,AFJ+08}, and ground based~\cite{SFI+11,CZKM,Elt02} quantum networks. The possible applications of such networks are not restricted to only cryptography. %Once quantum computers are available, 
A fully general quantum network will allow us to perform general distributed quantum computing~\cite{BBG+13,LB12,BKB+12}.

In this work, we study problems related to initialisation and construction of quantum networks. More specifically, we study how well $n$ nodes in an asynchronous quantum network can agree on a reference frame in the presence of at most $t$ arbitrarily faulty nodes among them. By asynchronous network we mean that in this setting we do not require the nodes to share a clock to  begin with, and the channel delays might vary arbitrarily in each use. In fact, an asynchronous protocol only assumes any message sent from a correct node to a correct node will eventually reach the destination, without imposing any bound on the channel delay. This assumption captures the most general reference frame agreement problem in a quantum network because during the initialisation of the network the pairwise channel delays might be unknown, clocks might not be synchronised and spatial reference frames might be unaligned. 

In a quantum channel, the qubits are encoded in some physical degree of freedom. For example, polarisation direction of photon is often used to encode qubits. This requires the sender and receiver to agree on some set of orthonormal directions as their common spatial reference frame. Another example is the time-bin qubits, where both of the parties require synchronised clocks. That is, they must have a pre-agreed temporal reference frame. 

So far these reference frame agreement problems are studied in a bipartite setting~\cite{MP95, PS01, BBM04, CD04, BM06, GLM06, SG12} with the exception of~\cite{IMSW14}, where spatial direction are agreed on in a \emph{synchronised} network of $n$ nodes. More specifically in~\cite{IMSW14} it is assumed that the network is synchronous. That is, all the nodes of the network have a shared clock and all the link delays have known upper bound. The bipartite reference frame agreement problem have been studied extensively (see~\cite{BRS07} for a review).  However, agreeing  on a reference frame in an asynchronous network of $n$ nodes remained open. 

There are protocols that allow Bell inequality tests and quantum information exchange between nodes without a pre-shared reference frame (See, for example~\cite{PVL+12,BCB13,DNW+12}. However, the ability to reliably share reference frames among multiple nodes gives significant technological advantages by simplifying the implementation of most protocols. Moreover, reference frame agreement protocols have important implications in fields that are not directly related to quantum information. 

One advantage of having an asynchronous reference frame agreement protocol for a network with certain number of faulty nodes is that once a spatial reference frame is established, then new robust protocols can potentially be built on top of it to perform network-wide clock synchronisation. This is a task important by itself with various applications in security, navigation and finance~\cite{KKB+14}. The primary difficulty of executing any protocol in an asynchronous network comes from the fact that in the presence of incorrect, that is, arbitrarily faulty nodes it is impossible to decide for a correct receiver whether a message is not arriving because the sender is faulty and not sending anything at all, or the sender is correct but the channel is taking a very long time to transfer the message. Therefore, it is nontrivial to decide how long to wait for a message before moving on to the next step of a protocol.

Another difficulty is that unlike in classical information theory where information can be represented in bits, a reference frame can only be transferred from scratch by exchanging systems which have an inherent sense of direction~\cite{PS201}. Examples of such systems are spin qubits and photon polarisation qubits.   The receiver can extract direction information from these systems, for example, by performing tomography on them. While preparing the direction any node $P_i$ will know the description of the direction as a vector $v_i$ in its local frame. Once the quantum system carrying that direction arrives at a receiver $P_j$, the receiver constructs a representation of the direction in it's own local frame as $v_j$. Such an estimation procedure inevitably introduces some error even in correct transmissions. That is, depending on the precision of the instruments one can only expect to have $d(v_i,v_j)\leq \delta$ for some $\delta >0$, where $d(v_i, v_j)$ is the Euclidian distance between $v_i$ and $v_j$. However, this distance metric does not make sense as it is, because $v_i$ and $v_j$ are vector representations in two different local frames. So we must redefine our distance metric $d(.,.)$ where distance is computed by converting both vectors in the frame of the first argument. As a result $d(v_i,v_j)$ remains a valid distance measure even though $P_i$ and $P_j$ do not know each other's local frame. This computation of distance between two vectors of different reference frames is only done in the analysis of the protocol and not by the nodes while playing the protocol. Any distance computed by a node inside a protocol is only between vectors for which it has a representation in its local frame. This inherent imperfection of message transmission must be accounted for by any reference frame agreement protocol. We capture this in the definition as, 

\begin{definition}
For $\eta >0$, a protocol in an asynchronous network of $n$ nodes is an \emph{$\eta$-asynchronous reference frame agreement protocol} if it satisfies the following conditions.
\vspace{-5pt} 
\begin{description}
	\setlength{\itemsep}{0pt}
	\item[Termination.] Every correct node $P_i$ eventually terminates and outputs a direction $v_i$. 
	\vspace{-5pt} 
	\item[Correctness.] If correct node $P_i$ outputs $v_i$ and correct node $P_j$ outputs $v_j$ then $d(v_i,v_j) \leq \eta$. 
\end{description}
\end{definition}
\vspace{-3pt} 
However, we have to achieve these termination and correctness condition in the presence of incorrect or faulty nodes. As it is unknown which nodes are faulty this resembles the Byzantine fault tolerance model~\cite{LSP82} studied in classical distributed computing. For quantum networks our assumptions are,
\vspace{-5pt} 
\begin{itemize}%[topsep=0pt]
	\setlength{\itemsep}{0pt}
	\setlength{\topsep=0pt}
	\item The pairwise channels are \emph{public}. That is, the messages are not secret. As a result, an adversary can see the content of a message between two correct nodes and adapt its strategy accordingly. 
	\item The pairwise channels are authenticated. That is, if a correct node sends a message to another correct node the message cannot be altered by any adversary. However, there might be channel noises, which can be dealt with, as in~\cite{IMSW14}.
	\item The pairwise channel delays might be controlled by the faulty nodes. That is, the faulty nodes can control the channel delays, even the delays for message passing between any pair of correct nodes.  
	\item If a correct node sends a message to another correct node, then the message eventually reaches the receiver. That is, even though the delay is controlled by some adversaries they cannot put infinite delay on the message between two correct nodes. However, the delay can be arbitrarily large.
	\item The faulty nodes might have correlated error. To create a protocol which tolerates the worst kind of faults, we also assume that the faulty nodes can cooperate with each other and have a global strategy to thwart the protocol. This is a realistic assumption because some nodes in a region might show correlated error which affects a part of the network. 

\end{itemize}
\vspace{-3pt} 
Under all these assumptions we give an \emph{$\eta$-asynchronous reference frame agreement protocol} for a network of $n$ nodes that can tolerate up to $t<n/4$ faulty nodes. We review some preliminaries before presenting the main results.

%Before we present our results we first review some preliminaries. 

%%%%%%%%%%%%%%%%%%%%%%%
\input{preliminaries}
%%%%%%%%%%%%

\section{Results}
In this paper we give a protocol that can take any two-party reference frame agreement protocol and lift it up to a fault tolerant multiparty 
reference frame agreement protocol. More specifically, we present the first protocol \pname{A-Agree} which allows $n$ nodes in a fully connected asynchronous quantum network to agree on a reference frame in the presence of $t < n/4$ faulty nodes. The result can be summarised in the following theorem. 
%\vspace{-5pt} 
\begin{theorem} \label{thm:aagree}
In a complete network of $n$ nodes that are pairwise connected by public authenticated quantum and classical channels, if a bipartite $\delta$-estimate direction protocol that uses $m$ qubits to achieve success probability $\qsucc \geq 1-e^{-\Omega(m{\delta^2})}$ is used, then protocol \pname{A-Agree} is a $42 \delta$-asynchronous reference frame agreement protocol with success probability at least $1-e^{-\Omega(m{\delta^2}-\log n)}$, that can tolerate up to $t<n/4$ faulty nodes.
 \end{theorem}

Note that, here we use the $\Omega$ notation. Therefore, the bounds on success probability 
asymptotically holds for large enough $m$. This is not a drawback because, for example, 
where photon polarisation is used to carry directional information, the pulses of polarised 
light created by the source would contain large number of photons and allow the protocol 
to achieve high success probability for a network of an arbitrary size.
%Note that, here we use the $\Omega$ notation. Therefore, the bounds on success probability asymptotically holds for large enough $m$. This is not a drawback because, for example, where photon polarisation is used to carry directional information the pulses of polarised light created by the source would contain large number of photons and allow the protocol to achieve high success probability for a network with any number of nodes.

The problem of both synchronous and asynchronous agreement on classical bits in the presence of arbitrarily faulty node is extensively studied in classical literature as Byzantine agreement problem \cite{LSP82}.  However, we emphasise that a classical protocol cannot be used in our problem because firstly, unlike classical network, any communication of direction among correct nodes in a quantum network will have inherent noises. As a result any classical protocol would see all the correct nodes as faulty nodes and the protocol will fail. Secondly, one cannot use the classical protocol directly because one cannot represent a reference frame using only classical bits~\cite{PS201}. However, classical literature can still inform us on important questions such as, how to achieve constant expected time, how to handle asynchronicity. Some of the approaches of our protocol regarding these questions are influenced by~\cite{CR93}. We also use the interactive consistency protocol by Ben-Or et al.~\cite{BE03} as a subroutine.

% Our protocol uses a subroutine called \pname{AR-Cast}
 %Now we will give the protocol. For that lets first define a few notations for convenience.
Before giving the protocols we first need to define some notation.

$w_i[j]$ represents a vector received by node $P_i$ from node $P_j$ using the bipartite direction estimation protocol. This vector is represented with respects to $P_i$'s local reference frame.

In our protocol sending $(\text{type},v)$ to some node means the sender  uses a \emph{$
\delta$-estimate direction protocol} to send the direction $v$ to the receiver. The sender also sends the classical tag [type] associated to this direction. The receiver will receive an approximation of the sent direction as $v'$ where $d(v,v')\leq \delta$. Our protocol uses four different tags as types. They are, init, echo, \readyone~and \readytwo.

Next, we fix a notation for a cluster of vectors of certain types where the cluster has a certain cluster centre, which is the average of the vectors, and a cluster parameter. We write it as $C_{i}^\delta(\mbox{[types]} , w_c)$. This means the cluster with cluster centre $w_c$ is computed and stored by node $P_i$, has a cluster parameter $\delta$ and contains only the vectors with associated tags in [types]. Here, [types] is a comma separated list of [type]s. The cluster parameter $\delta$ denotes that for all $u, v \in C_{i}^\delta(\mbox{[types]} , w_c)$ the distance $d(u, v) \leq \delta$.
%$C_{i}^\delta(\mbox{[types]} , w_c)$ denotes a set of vectors $w_i[j]$ of [type] messages that forms a cluster or radius with cluster centre $w_c$. That is, for all $w_i[j], w_i[k]\in C_{i}^\delta(w_c)$, $d(w_i[j],w_i[k]) \leq \delta$. If multiple types of vectors are clustered together, then a comma separated list is used as [type].

For example, $C_{i}^\delta(\mbox{[\readyone,\readytwo]} , v_c)$ denotes a cluster in which each vector has tags \readyone~or \readytwo~with cluster centre $v_c$ such that $\forall u, v \in C_{i}^\delta(\mbox{[\readyone,\readytwo]},v_c)$, and $d(u,v) \leq \delta$. 

$P(C_{i}^\delta(\mbox{[type]}, w_c))$ is the set of all the nodes $P_j$ such that, $w_i[j] \in C_{i}^\delta(\mbox{[type]}, w_c)$. That is, it is the set of node id's from which $P_i$ have received the vectors in the cluster $C_{i}^\delta(\mbox{[type]}, w_c)$.

Now we give our protocol in two steps. First, we give a protocol for asynchronous broadcast, which allows any sender to securely send a direction to all the other nodes. However, if the sender is faulty the protocol might never terminate. Using this as a primitive we later give our asynchronous agreement protocol. 

\subsection{Asynchronous broadcast}
As the name suggests using this protocol a sender node can send some message to all the other nodes in an asynchronous network. At first sight a naive protocol of just sending the message to all other nodes one by one seems to be a valid protocol. However, this naive protocol does not work if the sender intentionally sends different message to different nodes, which can easily happen in networks with faulty nodes. To guard from it, all the other nodes must communicate between each other to make sure they are receiving the same message, or a close approximation to it. However, as we have at most $t$ faulty nodes, this verification also becomes tricky. The whole thing becomes more challenging because the network is not synchronous. As a result a receiver who is waiting for a message, cannot be certain whether to keep waiting (because the message might be taking a long time in the channel) or move on (the sending node might be faulty and not sending the message at all). Our protocol takes care of all these challenges. 

Formally the protocol is defined as, 
\vspace{15pt}
\begin{definition} \label{def:acra}
For $\eta >0$, $\zeta>0$, a protocol which is initiated by a sender node $P_s$, in an asynchronous network of $n$ nodes, is called a $(\eta,\zeta)$-\emph{asynchronous reference frame broadcast protocol} if it satisfies the following conditions.
\begin{description}
	\item[Termination.]\ %If the sender is correct, each correct player $P_i$ eventually terminates the protocol, and outputs a reference frame $v_i$.
		\begin{enumerate}
			\item If the sender is correct then every correct node eventually completes the protocol.
			\item If any correct node completes the protocol, then all the correct nodes eventually complete the protocol.
		\end{enumerate}
	\item[Consistency.] If one correct node $P_k$ outputs a direction $v_k$ then all pairs of correct nodes $P_i$ and $P_j$ eventually output directions $v_i$, $v_j$ where $d(v_i,v_j) \leq \eta$. 
	\item[Correctness.] If $P_s$ is correct and broadcasts a direction $u$ and if a correct node $P_i$ outputs $v_i$ then \mbox{$d(u,v_i) \leq \zeta$}.
\end{description}
\end{definition}

We emphasize that the Termination condition of \emph{asynchronous reference frame  \textbf{broadcast}} is much weaker than the Termination condition of \emph{asynchronous reference frame \textbf{agreement}} because in the broadcast protocol we do not require that the correct nodes complete the protocol if the sender is faulty. Also, in an agreement protocol there is no designated sender node, whereas the broadcast protocol has a sender node.

We achieve asynchronous broadcast by our protocol \mbox{\pname{AR-Cast}}.  The following theorem summarises its properties. 

\begin{theorem} \label{thm:arc}
In a complete network of $n$ nodes that are pairwise connected by public authenticated classical and quantum channels, if a bipartite $\delta$-estimate direction protocol that uses $m$ qubits to achieve success probability $\qsucc \geq 1-e^{-\Omega(m{\delta^2})}$ is used, then protocol \pname{AR-Cast} is a $(42\delta,14\delta)$-\emph{asynchronous reference frame broadcast protocol}, with success probability at least $1-e^{-\Omega(m{\delta^2}-\log n)}$ that can tolerate up to $t<n/4$ faulty nodes.
\end{theorem}

\begin{algorithm}
\SetAlgorithmName{Protocol}{protocol}{List of Protocols}
	\LinesNumbered
	\DontPrintSemicolon
	\SetKwInOut{Input}{input}\SetKwInOut{Output}{output}
	\SetKwIF{Wait}{OrUntil}{Or}{Wait until}{Then}{\emph{Or} until}{Or}{}
	\SetKw{goto}{Goto}
	\SetKw{and}{And}
	\SetKw{StoA}{Send-to-all}
	\Input{Sender inputs direction $u$}
	\Output{$\forall i$ $P_i$ outputs direction $v_i$ }

\SetKwBlock{SEND}{Epoch~0: (Only Sender)}{}
\SEND{
	\StoA (init, $u$). 
}
\SetKwBlock{STEP}{Epoch~1: (Player $P_i$)}{}
\setcounter{AlgoLine}{0}
\STEP{
	Listen to init, echo, \readyone~and \readytwo~type messages.\\ 
	\Wait{\textbf{Either} received one (init, $u_i$)}{
		\StoA (echo, $u_i$).\; \goto Epoch~2.
	}
	\OrUntil{received a cluster of directions  $C_{i}^{4\delta}(\mbox{[echo]},w_c)$ of size at least $(n-2t) $ \and a cluster of directions $C_{i}^{10\delta}(\mbox{[\readyone,\readytwo]},v_c)$ of size at least $(t+1)$, so that, $d(w_c,v_c)\leq 10\delta$}{
	 \StoA (\readytwo, $w_c$).\;
	\goto Epoch~3.
	}
}

\SetKwBlock{STEP}{Epoch~2: (Player $P_i$)}{}
\setcounter{AlgoLine}{0}
\STEP{
	Listen to echo, \readyone~and \readytwo~type messages. \;
	\Wait{\textbf{Either} there exists a cluster of directions $C_{i}^{4\delta}(\mbox{[echo]},w_c)$ of size at least $(n-t)$}{
		\StoA (\readyone, $w_c$).\label{rdy1gen}\;
		\goto Epoch~3.
	}
	\OrUntil{there exists a cluster of directions  $C_{i}^{4\delta}(\mbox{[echo]},w_c)$ of size at least $(n-2t) $ \and a cluster of directions $C_{i}^{10\delta}(\mbox{[\readyone,\readytwo]},v_c)$ of size at least $(t+1)$, so that, $d(w_c,v_c)\leq 10\delta$,}{
	 \StoA (\readytwo, $w_c$).\; 
	 \goto Epoch~3.
	}
}

\SetKwBlock{STEP}{Epoch~3: (Player $P_i$)}{}
\setcounter{AlgoLine}{0}
\STEP{
\Wait{there exists a cluster of directions $C_{i}^{20\delta}(\mbox{[\readyone,\readytwo]},v_c)$ of size at least $(n-t)$ }{
		Output $v_c$. \;
		Halt\;
		}
} 

\caption{\pname{AR-Cast}} \label{async:broadcast}
\end{algorithm}

The protocol~\ref{async:broadcast}: \pname{AR-Cast} works roughly as follows. In Epoch 0 the sender sends its intended direction to all as a [init] type message. In Epoch 1 all the nodes wait until they receive an [init] from sender or sufficient number of confirmations from other nodes that they have received some directions and proceed to the next epoch. This way, even if some correct node never receives an [init] message, if the other correct nodes are advancing through the protocol, then this node in Epoch 1 will not stay behind waiting. In Epoch 2 the correct nodes, which have decided upon a direction, notify the other nodes about their decision by sending \readyone~ or \readytwo~ type messages to all. All these previous epochs make sure that all the correct nodes eventually arrive at Epoch 3 and outputs a direction which satisfies theorem~\ref{thm:arc}.
The formal proofs are given in the appendix. 

\subsection{Asynchronous agreement}

Now we give our main protocol~\pname{A-Agree} which uses~\mbox{\pname{AR-Cast}} as a subroutine and allows the correct nodes in an asynchronous network to agree on a reference frame.

\begin{algorithm}

\SetAlgorithmName{Protocol}{protocol}{List of Protocols}
	\LinesNumbered
	\DontPrintSemicolon
	\SetKwInOut{Input}{input}\SetKwInOut{Output}{output}
	\SetKwInOut{Input}{input}\SetKwInOut{Output}{output}
	\SetKwIF{Wait}{OrUntil}{Or}{Wait until}{Then}{Or until}{Or}{}
	\SetKw{goto}{Goto}
	\SetKw{and}{And}
	\SetKw{StoA}{Send-to-all}
	\Input{$\forall i$, $P_i$ inputs direction $u_i$}
	\Output{$\forall i$, $P_i$ outputs direction $v_i$ }

\SetKwBlock{STEP}{Epoch~0: (Player $P_i$) }{}
\STEP{

	Create a direction array $w_i$ of size $n$.\;
	$\forall j$, initialize $w_i[j] \leftarrow \perp$.\;
	Run \pname{AR-Cast}($u_i$).\; 	\label{aarunacast}
	\tcp{everyone broadcasts their local input}
	Store received direction from $P_j$ in $w_i[j]$.\;
	After receiving $(3t+1)$ such directions \goto Epoch~1. However, still continue the incomplete \pname{AR-Cast}s in parallel.\;
	%k = CC($\epsilon/4$). 
}

\SetKwBlock{STEP}{Epoch~1: (Player $P_i$)}{}
\setcounter{AlgoLine}{0}
\STEP{

	Create a bit string $a_i$ of size $n$.\;
	\For{ $j \leftarrow 1$ \KwTo $n$ }{
		\If{$w_i[j] \neq \perp$}{
	 		Assign $a_i[j] \leftarrow 1$.
		 }
		\Else{
			Assign $a_i[j] \leftarrow 0$.
		}
	}\label{AAE1a}}
	\tcp{$a_i$ records which \pname{A-Cast}s are completed so far by $P_i$}
	Run \pname{Asynchronous-IC}($a_i$).\; 
	\tcp{This step reports to all which \pname{A-Casts} are successfully received by $P_i$}
	Store the output of \pname{Asynchronous-IC} in vector $b_i$ such that, element $b_i[j]$ is received from $P_j$.\; \label{AAE19}
	\tcp{After this step every correct nodes know which \pname{A-Cast}s are reported to be complete by which node}
	\Wait { \pname{Asynchronous-IC} completes} {
	\goto Epoch~2}
	%After termination of \pname{Asynchronous-IC} \goto Epoch~2.

\SetKwBlock{STEP}{Epoch~2: (Player $P_i$)}{}
\setcounter{AlgoLine}{0}
\STEP{
	%Scan from column $1$ to $n$ and let $k_i$ to be the first column of $b_i$ that contains at least $(t+1)$ $1$s.\; \label{AAE2scan}
	
	%Let $k_i$ be the the chronologically earliest column index of $b_i$, such that it has at least $(t+1)$ `$1$'s in it.\;  \label{AAE2scan}
	Let $k_i$ be the index of a column which has at least $(t+1)$ 1s in it. So that, for any other index $l$ of column with (t+1) 1s  $k<l$.
	\tcp{After completion of \pname{Asynchronous-IC} each row of $b_i$ is a bit string of length $n$. That is $b_i$ is essentially an $n\times n$ bit matrix.} \label{AAE2scan}
	%\tcp{Here, $b_i$ is essentially an $n\times n$ array of bits with $b_i[j]$s as rows.} 
	\Wait{the \pname{A-Cast} initiated by $P_{k_i}$ completes}{ \label{AAE2wait}
	Assign $v \leftarrow w_i[k_i]$.\; \label{AAE2asg}
	Abort all incomplete \pname{A-Cast}s that are running since Epoch~0.\; \label{AAE2abort}
	Output $v$.
	}
}

\caption{\pname{A-Agree}} \label{async:agreement}
\end{algorithm}

In Epoch 0 of protocol~\ref{async:agreement}: \pname{A-Agree} each of the nodes $P_i$ proposes a direction $u_i$ ,which represents their local frame. They broadcast this direction using \pname{AR-Cast}. All the correct nodes wait for at least $(3t+1)$ such broadcasts to be complete. Then they enter Epoch 1. Since, there are $(3t+1)$ correct nodes they will eventually arrive at Epoch 1. In this step all the correct nodes create a bit string of length $n$ where $j$'th bit represents if the $j$'th \pname{AR-Cast} has been completed successfully in Epoch 0. Then all the nodes send this bit string to all by playing \pname{Asynchronous-IC}. After this they enter Epoch 2. In this Epoch every node has the same set of bit strings. They now look for the lowest inter $k$ such that at least $(t+1)$ bit strings have a $1$ in the $k$'th index of the string. If they have completed that $k$'th \pname{AR-Cast} they output their direction received from that broadcast. If the $k$'th \pname{AR-Cast} is not complete for a node, it waits until it completes and then output.
The election of $k$ ensures that at least one correct node has completed the $k$'th \pname{AR-Cast} so by Consistency of asynchronous reference frame broadcast all the correct nodes will eventually complete the $k$'th \pname{AR-Cast}. This ensures that the \pname{A-Agree} eventually completes. There is no conditional loop in this protocol and all the subroutines run in constant expected time. So, the \pname{A-Agree} is also a constant expected time protocol. 
The formal proofs are given in the appendix. 

\vspace{15pt}
\section{Conclusion}
In this work we have presented the first asynchronous reference frame agreement protocol. The synchronous protocol for spatial reference frame agreement presented in~\cite{IMSW14} can tolerate up to $t<n/3$ faulty nodes. Whereas, the asynchronous protocol we have presented tolerates only $t<n/4$ faulty nodes. Even though we pay this extra price in fault tolerance, an asynchronous protocol is a fully general reference frame agreement protocol. If message delays are fixed, our protocol can also be used to synchronise clocks~\cite{Chuang06}, which is an important problem in its own right. There are classical protocols for asynchronous agreement on bits which achieve $t<n/3$ in constant expected time, it remains open to see if this bound can be achieved by reference frame agreement protocols for a quantum network.

\begin{acknowledgments}
We thank Lo\"ick Magnin and Michael Ben-Or (via Loick Magnin) for helpful pointers, and David Elkouss for comments on an earlier version of this article. This work was supported by NRF CRP Grant ``Space based QKD''
and STW, QuTech. Stephanie Wehner is also supported by NWO VIDI Grant.
\end{acknowledgments}

%\bibliography{AAgree}

\clearpage

\input{arcast}

\input{aagree}

%%%%%%%%%
%\input{aic}

%\input{clock}

\end{document}

%% file: preliminaries.tex
\section{Preliminaries}

%In this section we introduce various concepts that are used throughout this work. 
The problem of reference frame agreement over an asynchronous quantum network is necessarily multidisciplinary in nature. That is, it combines various concepts from quantum physics, information theory, cryptography and distributed computing. In this section we introduce several concepts from these fields that will be useful throughout this work.

\subsection{Reference frame}
%We begin with the concept of \emph{reference frames}. In physics or in information theory all the quantities makes sense only in respect to some pre-established reference frame. Devices or parties that are communicating to each other either have the reference frames built in---for example this is how devices know whether the incoming signal is a `$0$' or a `$1$'---or have to establish it---for example devices might want to establish a clock rate or a reference direction---before meaningful communication can be started. Many quantum communication~\cite{}, quantum error correction~\cite{}, and quantum computing~\cite{} protocols leverage on the fact that all the parties involved share a common reference frame. These reference frames can be spatial, and temporal. 

\subsubsection{Spatial reference frame}
A \emph{spatial reference frame} defines a co-ordinate system in space. For example in a Cartesian coordinate system, once the Cartesian frame $(\vec{\bm{x}},\vec{\bm{y}},\vec{\bm{z}})$ is specified any vector $v = \alpha \vec{\bm{x}} + \beta \vec{\bm{y}} + \gamma \vec{\bm{z}}$ can be represented as $(\alpha, \beta, \gamma)$ where $\alpha, \beta$ and $\gamma$ are scalers. For two distant parties, who only have the knowledge of their own local frame, it becomes necessary to establish a shared reference frame before they can successfully communicate spatial information (such as, location and orientation). 

We use quantum communications to send a direction between a sender and a receiver. Any protocol that allows transmission of direction between two nodes with $\delta$ accuracy is called a  2-party $\delta$-estimate direction protocol. As an example we refer to the Protocol~\ref{prt:2ed}, \pname{2ED}, one of the simplest possible protocols as studied in~\cite{MP95}. Here a sender creates many identical qubits with their Bloch vector pointing to the intended direction and the receiver measures them with Pauli measurements. From the statistics of the measurement outcomes, the receiver then estimates the Bloch vector's direction within Euclidian distance $\delta$ with probability of success  $\qsucc \geq 1-e^{-\Omega(\delta^2m)}$ where $m$ is the number of qubits exchanged. That is, the Protocol \pname{2ED} allows the sender to transmit a direction $u$ which is received as the direction $v$ at the receiver, such that the inequality $d(u,v) \leq \delta$ holds with probability $\qsucc \geq 1-e^{-\Omega(\delta^2m)}$. We emphasise that, this work allows us to lift any two party $\delta$-estimate direction protocol into a protocol for a quantum network of $n$ nodes.

\begin{algorithm}
\SetAlgorithmName{Protocol}{protocol}{List of Protocols} 
	\LinesNumbered
	\DontPrintSemicolon
	\SetKwInOut{Input}{input}\SetKwInOut{Output}{output}
	\Input{Sender, direction $u$}
	\Output{Receiver, direction $v$ }

\SetKwBlock{SEND}{Sender: \pname{2ED-Send}}{}	
\SEND{
	Prepare $3n$ qubits with direction $u$ \; 
	Send them to the receiver

}
\SetKwBlock{RECIEVE}{Receiver: \pname{2ED-Receive}}{}	
\RECIEVE{
	Receive $3n$ qubits from the sender\;
	Measure $n$ qubits with $\sigma_x$ and compute $p_x$, the frequency of getting outcome~$+1$ \;
	Similarly on the remaining qubits, compute $p_y$ and $p_z$ with measurements $\sigma_y$ and $\sigma_z$ on $n$ qubits each \;
	Assign $x \leftarrow 2 p_x-1$, $y \leftarrow 2 p_y-1$, $z \leftarrow 2 p_z-1$;\ 
	Assign $l \leftarrow \sqrt{x^2+y^2+z^2}$\;
	Output $v \leftarrow (x/l,y/l,z/l)$\;
}
\caption{\pname{2ED}} \label{prt:2ed}
\end{algorithm}

%For example in a multi player setting if player $P_i$ has a local Cartesian frame $(\vec{\bm{x}}_i,\vec{\bm{y}}_i,\vec{\bm{z}}_i)$, he can represent a vector $v_i = \alpha_i \vec{\bm{x}}_i + \beta_i \vec{\bm{y}}_i + \gamma_i \vec{\bm{z}}_i$ as $(\alpha_i,\beta_i,\gamma_i)$. But without the knowledge of $P_i$'s local frame another player $P_j$ cannot make sense of the vector $v_i=(\alpha_i,\beta_i,\gamma_i)$. For an effective communication of vectors, (i.e. location and direction information) they need a pre-established shared reference frame. 

\subsubsection{Temporal reference frame}

Similar to spatial reference frames multiple parties might need to synchronise their clock rates and time differences. Once they have established it, we say that they share a \emph{temporal reference frame} and they are synchronised in time. Any multiparty protocol or computation performed by systems that do not share a temporal reference frame are respectively called \emph{asynchronous protocol} or \emph{asynchronous computation}.

\subsection{Asynchronous communication}
In an asynchronous network we assume that the nodes do not share any synchronised clock. And the communication channel between each pair is such that a message takes an arbitrary amount of time to propagate through it. Here the only guarantee is, if a message is transmitted from a correct node the message will eventually reach to the receiver. Also, a node might take an arbitrary amount of time to perform the next step in a protocol. In this setup, to analyse the time complexity of an asynchronous protocol we only count the maximum number of steps executed by any node before the protocol completes, and call it the running time of the protocol. The performance, in terms of execution time, of an asynchronous agreement protocol is determined by its expected running time. The expectation is thereby taken over all possible random inputs of the nodes, random bits used by the nodes, as well as all possible random behaviour of the faulty nodes. The exact probability distributions may not be known, but the goal is to show that the expected running time is low for all possible distributions.

\subsubsection{The asynchronous message}

%\subsection{Semi-synchronous communication}
%In a semi-synchronous network, we assume that all the nodes have their own clock, but the clocks are not synchronized. Also any message takes arbitrary but bounded amount of time to reach to the receiver. We denote this maximum rely by $d_c$.

In the absence of a synchronised clock, each message must have a `begin' and `end' tag. Also, depending on the particular application, a message might carry a [type] tag. In our problem we don't have a shared reference frame. For this reason, we cannot use the quantum channel to carry these [type] tags. This requires us to have a parallel classical channel that uses some classical degree of freedom to carry bits.  

We assume that each pair of nodes are connected by an asynchronous public authenticated CQ-channel (classical quantum channel), which can send a message using both classical and quantum degrees of freedom in the absence of a shared reference frame. An example of such combined message is shown in Table~\ref{table:msg} where each quantum message $m_q$ is sandwiched between a classical `begin' and an `end' tag and also accompanied by a classical type tag $m_c$. The symbol $\perp$ denotes quantum signals that can be ignored. 

\begin{table}[ht]
\caption{Channel primitive: \bf{A message}}% title of Table
\centering % used for centering table
\begin{tabular}{c c c } % centered columns (4 columns)
\hline\hline %inserts double horizontal lines
Step &Classical & Quantum \\ [0.5ex] % inserts table
%heading
\hline % inserts single horizontal line
1 & begin & $\perp$  \\ % inserting body of the table
2 & $m_c$ & $m_q$  \\
3 & end & $\perp$  \\ [1ex] % [1ex] adds vertical space
\hline %inserts single line
\end{tabular}
\label{table:msg} % is used to refer this table in the text
\end{table}

The only assumption is the nodes can match the classical and quantum parts of the message.

%Messages might come out of a channel in a different order than the were put in initially. But we assume inside one atomic message the order of each of the 3 steps remains unchanged.

%\subsubsection{Time complexity of asynchronous protocols}
%
%Note that in an asynchronous network a message can take arbitrary amount of time to travel from sender to receiver. Also in asynchronous protocol, the nodes might take arbitrary amount of time to take next step in a protocol. So, to compute the time complexity of an asynchronous protocol it is a standard practice to only count the  number of steps executed by any node. Each of these steps might take an arbitrary amount of time. 

\subsubsection{Asynchronous interactive consistency}
Our protocol uses the solution to the following interactive consistency problem which was first proposed by Pease, Shostak and Lamport~\cite{PSL80}.

\begin{definition}[The Interactive Consistency Problem]
Consider a complete network of $n$ nodes in which communication lines are private. Among the $n$ nodes up to $t$ might be faulty. Let $P_1, P_2, \ldots, P_n$ denote the nodes. Suppose that each node $P_i$ has some private value of information $V_i \in \abs{V} \geq 2$. The question is whether it is possible to devise a protocol that, given $n, t\geq 0$, will allow each correct node to compute a vector of values with an element for each of the $n$ processors, such that: 

\begin{enumerate}
\item All the correct nodes compute exactly the same vector;
\item The element of this vector corresponding to a  given correct node is the private value of that node.
\end{enumerate}

\end{definition}

For an asynchronous network, Ben-Or and El-Yaniv~\cite{BE03} gives a protocol \pname{Asynchronous-IC} which solves this problem for $t<n/3$ in constant expected time. We use this protocol as a subroutine. 

Not that the \pname{Asynchronous-IC} requires private asynchronous classical channels. Whereas, we only require public authenticated classical and quantum channels between each pair of nodes in the network. The reason is, with authenticated public quantum channels each pair of nodes can play \pname{2ED} type protocol and establish a bipartite reference frame. Once the bipartite reference frame is established between each pair using the public authenticated classical and quantum channels they can perform QKD which gives them a private classical channel. So, they can play \pname{Asynchronous-IC} at a later stage of the protocol. We emphasise that, even thought by playing pairwise \pname{2ED} each honest pair of nodes can share a reference frame between them the goal of this paper is to have a global shared reference frame which is non-trivial in the presence of faulty nodes. 

%\subsection{Private authenticated classical channel from public quantum channel}

%\subsubsection{Byzantine fault}
%\subsubsection{Correct nodes}
%\subsubsection{Faulty nodes}

%\subsection{Asynchronous protocol}

%\subsection{Clock synchronization from spatial reference frame alignment}

%% file: arcast.tex
\section{appendix}

\subsection{Asynchronous reference frame broadcast} 
\label{prot:AB}

To prove correctness of our \pname{AR-Cast} we have to prove theorem~\ref{thm:arc} as repeated here. 

\setcounter{theorem}{1}
\begin{theorem}  
In a complete network of $n$ nodes that are pairwise connected by public authenticated quantum and classical channels, if a bipartite $\delta$-estimate direction protocol that uses $m$ qubits to achieve success probability $\qsucc \geq 1-e^{-\Omega(m{\delta^2})}$ is used, then protocol \pname{AR-Cast} is a $(42\delta,14\delta)$-\emph{asynchronous reference frame broadcast protocol}, with success probability at least $1-e^{-\Omega(m{\delta^2}-\log n)}$ that can tolerate up to $t<n/4$ faulty nodes.
\end{theorem}

For this we observe several properties of Protocol~\ref{async:broadcast} in the following lemmas. The first observation is that if two different correct nodes send [\readyone]~type messages then the direction they send are close to each other with high probability. 

\begin{lemma} \label {lem:readyone}
For $t<n/4$, $\delta>0, \qsucc>0$, if two correct nodes $P_i$ and $P_j$ send ([\readyone],$u$) and ([\readyone],$v$) respectively, then $d(u,v) \leq 10 \delta$ with probability at least $\qsucc^{n+n^2}$.
\end{lemma}

\begin{proof}
In step~\ref{rdy1gen} of Epoch 2 when a [\readyone]~message is generated there are at most $n$ init messages originated from the sender and at most $n^2$ echo messages generated by the other nodes. So, with probability at least $\qsucc^{n+n^2}$ all the transmissions which are among correct nodes are successful. Conditioning on this, we prove, 
\begin{align}\label{eq:lemr1}
d(u,v) \leq 10 \delta. 
\end{align}

We show this in two steps. First, we show that there exists a common correct node $P_k$ in $P(C_{i}^{4\delta}(\mbox{[echo]},u))$ and $P(C_{j}^{4\delta}(\mbox{[echo]},v))$, where $C_{i}^{4\delta}(\mbox{[echo]},u)$ and $C_{j}^{4\delta}(\mbox{[echo]},v)$ are the cluster of echo type directions received by $P_i$ and $P_j$ respectively .
Then using the triangle inequality with the fact that the echo vector from $P_k$ must be close to both of the cluster centers $u$ and $v$, we derive inequality~\eqref{eq:lemr1}.

Now, for the first step, let us denote $A_i$ and $ A_j$ to be the set of nodes from which the vectors respectively in $C_{i}^{4\delta}(\mbox{[echo]},u)$ and $C_{j}^{4\delta}(\mbox{[echo]},v)$ have originated. And $B_i$ and $B_j$ to be the correct nodes in $A_i$ and $ A_j$ respectively. Formally, 

\begin{align}
A_i &= P(C_{i}^{4\delta}(\mbox{[echo]},u)),\\
A_j &= P(C_{j}^{4\delta}(\mbox{[echo]},v)),\\
B_i &= \{ P_l:P_l\in A_i \mbox{ and } P_l \mbox{ is correct.} \},\\
B_j &= \{ P_l:P_l\in A_j \mbox{ and } P_l \mbox{ is correct.} \}.
\end{align}

Note that at this step $\abs{A_i} \geq n-t$ and $\abs{A_j} \geq n-t$. We want to show that,
\begin{align}
B_i \cap B_j \neq \emptyset.
\end{align}

We do this by contradiction: let us assume that,

\begin{align}
B_i \cap B_j = \emptyset. \label{eq:cnt}
\end{align}

Note that, 
\begin{align}
&\abs{A_i} \geq n-t \\
&\Rightarrow \abs{A_i - B_i} + \abs{B_i} \geq n-t,\\ 
\label{eq:nonfaulty}&\Rightarrow t + \abs{B_i} \geq n-t,\\
&\Rightarrow \abs{B_i} \geq n-2t,\\
\label{eq:nonmain}&\Rightarrow \abs{B_i} > n - 2 (n/4) = n/2.
\end{align}

Here, inequality~\eqref{eq:nonfaulty} holds because at most $t$ of the nodes are faulty. And inequality~\eqref{eq:nonmain} holds because $t<n/4$.

Now, 
\begin{align}
\nonumber \abs{A_i \cup A_j} &= \abs{(A_i - B_i )\cup (  A_j - B_j) \cup B_i \cup B_j},\\
\label{eq:usecnt} &\geq  \abs{( A_j- B_j)}  +\abs{ B_i} +\abs{ B_j},\\ 
&= \abs{ A_j}  +\abs{ B_i},\\
\label{eq:lm1sub}&> (n-t) +n/2,\\
\label{eq:lem1final}& > n - n/4+n/2 = 5n/4
\end{align}
Here, inequality~\eqref{eq:usecnt} uses inequality~\eqref{eq:cnt}, inequality~\eqref{eq:lm1sub} follows from the definition from the size of $A_j$ and inequality~\eqref{eq:nonmain}. And inequality~\eqref{eq:lem1final} follows because, $t<n/4$.
However, this is a contradiction, because there are only $n$ nodes in the network. Therefore, we have, 
\begin{align}
B_i \cap B_j \neq \emptyset.
\end{align}

So, there exists a common correct node $P_k \in B_i \cap B_j$ in $P(C_{i}^{4\delta}(\mbox{[echo]},u))$ and $P(C_{j}^{4\delta}(\mbox{[echo]},v))$. Since $P_k$ is correct, it must have sent the same echo type message to both $P_i$ and $P_j$. So, using the triangle inequality we have,  
\begin{align}
d(w_i[k], w_j[k]) &\leq d(w_i[k],u_k) + d(u_k,w_j[k]),\\
\label{eq:r11correctdist}&\leq \delta + \delta = 2\delta.
\end{align}

Now inequality~\eqref{eq:lemr1} follows because,
\begin{align}
&d(u,v) \leq d(u,w_i[k])+ d(w_i[k],w_j[k])+d(w_j[k],v),\\
\label{eq:r11appldef} &\leq 4\delta + d(w_i[k],w_j[k]) + 4\delta, \\
\label{eq:r11correct}&\leq 4\delta + 2\delta + 4\delta = 10\delta.
\end{align}
Here, inequality~\eqref{eq:r11appldef} follows from the definitions of $C_{i}^{4\delta}(\mbox{[echo]},u)$ and $C_{j}^{4\delta}(\mbox{[echo]},v)$ and inequality~\eqref{eq:r11correct} follows from inequality~\eqref{eq:r11correctdist}.
\end{proof}

In lemma~\ref{lem:readyone} we have shown the relation between two [\readyone] type directions from two different correct nodes. Now we show that if a correct node sends a [\readyone] and another correct node sends a [\readytwo]~type message then the directions they send are close with high probability. Both of these proofs use similar techniques.

\begin{lemma} \label{lem:readyonetwo}
For $t<n/4$, $\delta>0, \qsucc>0$, if two correct nodes $P_i$ and $P_j$ send ([\readyone],$u$) and ([\readytwo],$v$) accordingly, then $d(u,v) \leq 10 \delta$ with probability at least $\qsucc^{n+2n^2}$.
\end{lemma}

\begin{proof}

When a [\readytwo]~message is generated there are at most $n$ init, $n^2$ echo and in total $n^2$ [\readyone]~or [\readytwo]~messages generated in the protocol. With probability at least $\qsucc^{n+2{n^2}}$ all the transmissions which are among correct nodes  are successful. Conditioning on this, we show that, 
\begin{align} \label{lemr12main}
d(u,v)\leq 10\delta.
\end{align}

We do this in two steps, first we show that there is a common correct node $P_k$ in $P(C_{i}^{4\delta}(\mbox{[echo]},u))$ and $P(C_{j}^{4\delta}(\mbox{[echo]},v))$. Then using the triangle inequality with the fact that both of the cluster centers $u$ and $v$ must be close to the echo direction sent from $P_k$ we prove the inequality~\eqref{lemr12main}.

Now, for the first step, let us denote $A_i$ and $ A_j$ to be the set of nodes from which the vectors respectively in $(C_{i}^{4\delta}(\mbox{[echo]},u)$ and $C_{j}^{4\delta}(\mbox{[echo]},v)$ have originated. And $B_i$ and $B_j$ to be the correct nodes in $A_i$ and $ A_j$ respectively. Formally, 

\begin{align}
A_i &= P(C_{i}^{4\delta}(\mbox{[echo]},u)),\\
A_j &= P(C_{j}^{4\delta}(\mbox{[echo]},v)),\\
B_i &= \{ P_l:P_l\in A_i \mbox{ and } P_l \mbox{ is correct.} \},\\
B_j &= \{ P_l:P_l\in A_j \mbox{ and } P_l \mbox{ is correct.} \}.
\end{align}

Note that here $\abs{A_i} \geq n-t$ and $\abs{A_j} \geq n-2t$. 
We want to show that,
\begin{align}
B_i \cap B_j \neq \emptyset.
\end{align}

We do this by contradiction: let us assume that,

\begin{align}
B_i \cap B_j = \emptyset.
\end{align}

Note that, 
\begin{align}
&\abs{A_i} \geq n-t \\
&\Rightarrow \abs{A_i - B_i} + \abs{B_i} \geq n-t,\\
\label{eq:r12nonfaulty}&\Rightarrow t + \abs{B_i} \geq n-t,\\
&\Rightarrow \abs{B_i} \geq n-2t,\\
\label{eq:r12nonmain}&\Rightarrow \abs{B_i} > n - 2 (n/4) = n/2.
\end{align}

Here, inequality~\eqref{eq:r12nonfaulty} holds because at most $t$ of the nodes are faulty. And inequality~\eqref{eq:r12nonmain} holds because $t<n/4$.

Now, 
\begin{align}
\nonumber \abs{A_i \cup A_j} &= \abs{(A_i - B_i )\cup (  A_j - B_j) \cup B_i \cup B_j},\\
&\geq  \abs{( A_j- B_j)}  +\abs{ B_i} +\abs{ B_j},\\
&= \abs{ A_j}  +\abs{ B_i},\\
\label{eq:r12sub}&> (n-2t) +n/2,\\
\label{eq:r12final}& > n - n/2+n/2 =  n
\end{align}

Here, inequality~\eqref{eq:r12sub} follows from the definition of $A_j$ and inequality~\eqref{eq:r12nonmain}. And inequality~\eqref{eq:r12final} follows because, $t<n/4$.
However, this is a contradiction, because there are only $n$ nodes in the network. Therefore, we have, 
\begin{align}
B_i \cap B_j \neq \emptyset.
\end{align}

So, there exists a common correct node $P_k$ in $P(C_{i}^{4\delta}(\mbox{[echo]},u))$ and $P(C_{j}^{4\delta}(\mbox{[echo]},v))$. As $P_k$ is correct, it must have sent the same echo type message to both $P_i$ and $P_j$. So, using the triangle inequality we have,  
\begin{align}
d(w_i[k], w_j[k]) &\leq d(w_i[k],u_k) + d(u_k,w_j[k]),\\
\label{eq:correctdist}&\leq \delta + \delta = 2\delta.
\end{align}

Now inequality~\eqref{lemr12main} follows because,
\begin{align}
&d(u,v) \leq d(u,w_i[k])+ d(w_i[k],w_j[k])+d(w_j[k],v),\\
\label{eq:appldef} &\leq 4\delta + d(w_i[k],w_j[k]) + 4\delta, \\
\label{eq:correct}&\leq 4\delta + 2\delta + 4\delta = 10\delta.
\end{align}
Here, inequality~\eqref{eq:appldef} follows from the definitions of $C_{i}^{4\delta}(\mbox{[echo]},u)$ and $C_{j}^{4\delta}(\mbox{[echo]},v)$ and inequality~\eqref{eq:correct} follows from inequality~\eqref{eq:correctdist}.

\end{proof}

Now we show that all the correct nodes cannot send only [\readytwo]~type messages. That is, if there exists a [\readytwo]~message sent from a correct node, then there must pre-exist a [\readyone]~message sent from another correct node. 

\begin{lemma} \label{lem:causal}
For $t<n/4$, $\delta>0, \qsucc>0$, if a correct node $P_j$ sends ([\readytwo],$v$), then with probability at least $\qsucc^{n+2n^2}$, there exists a correct node $P_i$ which has sent ([\readyone],$u$) .
\end{lemma}

\begin{proof}
When a [\readytwo]~message is generated there are at most $n$ [init], $n^2$ [echo] and in total $n^2$ [\readyone]~or [\readytwo]~messages generated in the protocol. With probability at least $\qsucc^{n+2{n^2}}$  all the transmissions which are among correct nodes are successful. In this case,
just before making the decision to send a ([\readytwo],$v$) message node $P_j$ must have received at least (t+1) [\readyone]~or [\readytwo]~messages from nodes in $P(C_{i}^{10\delta}(\mbox{[\readyone,\readytwo]}v_c))$. Of these, at least one node---let's call it $P_k$---is correct. If $P_k$ has also sent a [\readytwo]~type message, we can find another correct node in its $P(C_{k}^{10\delta}(\mbox{[\readyone,\readytwo]}v_c))$ and so on. This way, eventually we will find a correct node who has sent a [\readyone]~type message.

To see this, let us define a directed graph $G(V,E)$ with vertex set $V= \{P_i : P_i \mbox{ is correct}\}$,
and 
\begin{align}
E = \{&(P_k,P_i): P_k\mbox{ has sent ready}_2 \nonumber\\
&\mbox{ after receiving ready}_1\mbox{ or ready}_2 \mbox{ from } P_i\}.
\end{align}

%That is, a directed edge $(P_k,P_i)$ implies, $P_i$ sends a [\readytwo]~message upon receiving a [\readyone]~or [\readytwo]~message from $P_k$.
One can convince oneself that $G$ is a directed acyclic graph because any cycle in the graph would violate the cause and effect relation of the edge directions.  Now if we look at the connected component of this graph containing $P_j$ there must exist a node $P_i$ in this component with no outgoing edges. Because $V$ only contains correct nodes. This implies $P_i$ is a correct node which has sent a [\readyone]~type message ([\readyone],$u$). This completes the proof.
\end{proof}

Now the only thing that remains is to show that two [\readytwo]~type directions sent from two correct nodes are close with high probability. 

\begin{lemma} \label{lem:r22}
For $t<n/4$, $\delta>0, \qsucc>0$, if two nodes $P_i$ and $P_j$ sends ([\readytwo],$u$) and ([\readytwo],$v$) respectively, then $d(u,v) \leq 20 \delta$ with probability at least $\qsucc^{n+2n^2}$.
\end{lemma}

\begin{proof}
When a [\readytwo]~message is generated there are at most $n$ [init], $n^2$ [echo] and in total $n^2$ [\readyone]~or [\readytwo]~messages generated in the protocol. With probability at least $\qsucc^{n+2{n^2}}$  all of these transmissions which are between correct nodes are successful. Conditioning on this, we show that,
if correct $P_i$ sends ([\readytwo],$u$) then from Lemma~\ref{lem:causal} there exists a correct node $P_k$ which has sent ([\readyone],$w$). From Lemma~\ref{lem:readyonetwo}, 
\begin{align}
d(u,w) \leq 10\delta,
\end{align}
and 
\begin{align}
d(v,w) \leq 10\delta.
\end{align}

Using the triangle inequality with these we get,
\begin{align}
d(u,v)\leq d(u,w) + d(w,v) \leq 10\delta + 10\delta = 20\delta.
\end{align}

\end{proof}

%%%%%%%%%%%%

Now we are ready to prove that our protocol~\ref{async:broadcast} satisfies the first termination condition of definition~\ref{def:acra}.

\begin{lemma}[Termination 1]\label{lem:acastterm1}
For $t<n/4$, $\delta>0, \qsucc>0$, if the sender $P_k$ is correct then the protocol~\ref{async:broadcast} \emph{AR-Cast} eventually terminates with probability at least $\qsucc^{n+n^2}$.
\end{lemma}

\begin{proof}
There are at most $n$ [init] messages, $n^2$ [echo] messages and $n^2$ [\readyone]~or [\readytwo]~type messages exchanged in the protocol. With probability at least $\qsucc^{n+2n^2}$  all of these transmissions which are between correct nodes are successful.
In this case, if the sender is correct all the correct nodes eventually receive [init] messages that are at most $2\delta$ apart from each other and send an echo message. So, all the received [echo] messages are at most $3\delta$ apart from the received direction in the [init] message of any correct node. Any node that has sent a [\readyone]~type message will go to epoch 3. The faulty nodes cannot stop the [init] and [echo] messages from correct nodes but they can manipulate the delays, so that some of the correct nodes send [\readytwo]~type messages. After sending the [\readytwo]~these correct nodes will eventually arrive at Epoch 3. From lemma~\ref{lem:readyone} and lemma~\ref{lem:readyonetwo} we can see that for any correct $P_i$ all the received [\readyone]~and [\readytwo]~directions will be in $C_{i}^{16\delta}(\mbox{[\readyone,\readytwo]},v_c)$. And because there are $(n-t)$ of them originating from the correct nodes the protocol~\ref{async:broadcast} \pname{AR-Cast} will eventually terminate. Note that, if the sender is faulty, the definition of \emph{$(\eta,\zeta)$-reference frame broadcast protocol} (Derinition~\ref{def:acra}) do not require any termination. 
\end{proof} 

Now we show that if one correct node outputs a direction, then all the correct nodes eventually output directions that are close to each other. 

\begin{lemma}[Consistency] \label{lem:acastcon}
For $t<n/4$, $\delta>0, \qsucc>0$, in protocol \emph{AR-cast}, if a correct node $P_k$ outputs $v_k$  then all pair of correct nodes $P_i, P_j$ eventually output $v_i, v_j$ respectively such that,
$d(v_i,v_j) \leq 42\delta$ with probability at least $\qsucc^{n+n^2}$.
\end{lemma}

\begin{proof}
When a [\readytwo]~message is generated there are at most $n$ init, $n^2$ echo and in total $n^2$ [\readyone]~or [\readytwo]~messages generated in the protocol. With probability at least $\qsucc^{n+2{n^2}}$ all of these transmissions which are between correct nodes are successful. In this case, we prove, 
\begin{align} 
\label{eq:bcnst} d(v_i,v_j) \leq 42\delta,
\end{align}
by showing that the successful completion of $P_k$ implies there are enough echo, [\readyone]~and [\readytwo]~type messages generated by correct nodes so that all the other correct nodes eventually receive them and successfully terminate and each pair of their outputs satisfies inequality~\eqref{eq:bcnst}.

 Now, if a correct node $P_k$ outputs $v_k$ then this implies it has received at least $(n-t)$ [\readyone]~or [\readytwo]~messages from nodes in  $P(C_{k}^{20\delta}(\mbox{[\readyone,\readytwo]},v_k))$, of which at least $(n-2t)$ are correct. Messages from these correct nodes eventually reach all the other correct nodes. Also, from lemma~\ref{lem:causal} there exists a correct node which has sent a [\readyone]~ message which implies all the correct nodes eventually receive at least $(n-2t)$ echo messages. That is, all the correct nodes waiting in Epoch 1 or Epoch 2 will satisfy the condition of sending a [\readytwo]~message and go to Epoch 3. Any correct node $P_i$, $P_j$ waiting in Epoch 3 will eventually receive all the [\readyone]~or [\readytwo]~messages sent from correct nodes in $P(C_{i}^{20\delta}(\mbox{[\readyone,\readytwo]},v_i))$ and $P(C_{j}^{20\delta}(\mbox{[\readyone,\readytwo]},v_j))$ accordingly, and output $v_i$, $v_j$ accordingly. 
 
Now we show that $P(C_{i}^{20\delta}(\mbox{[\readyone,\readytwo]},v_i))$ and $P(C_{j}^{20\delta}(\mbox{[\readyone,\readytwo]},v_j))$ have at least one common correct node, which implies the cluster centers are close. 
 
To see this note that each of these clusters have at least $(n-2t)>n-2(n/4)=n/2$ correct nodes. That is more than $n$ correct nodes in total. However there are total $n$ nodes in the networks. This implies at least some of the correct nodes are common in both clusters. Let $P_l$ be such a node.

Now using triangular inequality we have, 
\begin{align}
d(v_i,v_j) &\leq d(v_i, v_i[l])+d(v_i[l],v_l)\nonumber\\&+d(v_l,v_j[l]) +d(v_j[l],v_j),\\
\label{eq:l42d}&\leq20\delta + \delta +\delta + 20\delta = 42\delta.
\end{align}

Here inequality~\eqref{eq:l42d} follows using lemma~\ref{lem:r22}.
\end{proof}

Now the second termination condition.
\begin{lemma}[Termination 2]\label{lem:acastterm2}
For $t<n/4$, $\delta>0, \qsucc>0$, if a correct node $P_i$ completes the protocol then all the correct nodes complete the protocol with probability at least $\qsucc^{n+2{n^2}}$.
\end{lemma}

\begin{proof}
This lemma is a corollary of lemma~\ref{lem:acastcon}. Because lemma~\ref{lem:acastcon} ensures completion with probability at least $\qsucc^{n+2{n^2}}$.

\end{proof}

Now we are ready to prove that our protocol satisfies the correctness condition of definition~\ref{def:acra}.

\begin{lemma}[Correctness] \label{lem:arcmeaning}
For $t<n/4$, $\delta>0, \qsucc>0$, if a correct sender $P_s$ sends (init,u) and a correct node $P_i$ outputs $v_i$ then $d(u,v_i)\leq 14\delta$ with probability $\qsucc^{n+2{n^2}}$.
\end{lemma}

\begin{proof}

There are at most $n$ init messages, $n^2$ echo messages and $n^2$ [\readyone]~or [\readytwo]~type messages exchanged in the protocol. With probability at least $\qsucc^{n+2n^2}$  all of these transmissions which are between correct nodes are successful.

In this case we prove the lemma in three steps. First, we show that all the [\readyone]~type directions sent from correct nodes are close to $u$. Secondly, we show that all the [\readytwo]~type directions sent from the correct nodes are close to $u$. And finally, from these we conclude the proof. 

For the first step, let us assume that correct node $P_i$ has sent a ([\readyone], $v_i$) message in Epoch 2. So, it has received at least $(n-t)$ echo type messages, of which at least $(n-2t)$ are from correct nodes. Let's assume for some correct node $P_j$ $w_i[j] \in C_i^{4\delta}(v_i)$. Since $P_j$ is correct, using the triangle inequality, we have, 

\begin{align}
d(u,w_i[j]) &\leq d(u,u_j)+d(u_j,w_i[j]),\\
		&\leq \delta + \delta = 2\delta. \label{eq:stocorrect}
\end{align}

The diameter of the cluster $C_i^{4\delta}(v_i)$  is $4\delta$. So, we have, $d(v_i,w_i[j]) \leq 2\delta$. Using this and~\eqref{eq:stocorrect} with the triangle inequality, we have, 
\begin{align}
d(u,v_i) 	& \leq d(u,w_i[j]) + d(w_i[j],v_i),\\
		& \leq 2\delta + 2\delta = 4\delta. \label{eq:stor1}
\end{align}

Now, for the second step, let us assume that a correct node $P_l$ has sent a ([\readytwo], $v_l$) message from Epoch 1 or Epoch 2. So, $v_l$ is a cluster center of at least $(n-2t)$ echo type messages. Of which at least $(n-3t)$ are correct. So, a similar reasoning to the previous step shows,

\begin{align}
d(u,v_l) \leq 4\delta. \label{eq:stor2}
\end{align}

Finally, since the sender is correct from lemma~\ref{lem:acastterm1} we know, all the correct nodes eventually enter Epoch 3 and successfully complete the epoch.

Let us assume a correct node $P_i$ has received a cluster of [\readyone]~or [\readytwo]~ type directions $C_{i}^{20\delta}(\mbox{[\readyone,\readytwo]},v_c)$ of size at least $(n-t)$. So, there is a correct node $P_k$ for which $v_i[k] \in C_{i}^{20\delta}(\mbox{[\readyone,\readytwo]},v_c)$. Here, $C_{i}^{20\delta}(\mbox{[\readyone,\readytwo]},v_c)$ is a cluster of diameter $20\delta$. So, we have $d(v_i[k],v_c)\leq 10\delta$. Using the triangle inequality with this, and~\eqref{eq:stor1} and~\eqref{eq:stor2}, we have, 

\begin{align}
d(u,v_c)	&\leq d(u,w_i[k]) +d(w_i[k],v_c),\\
		&\leq 4\delta + 10\delta = 14\delta.
\end{align}

This concludes the proof.

\end{proof}

Now we give an auxiliary lemma that shows how the probability of success scales with the number of nodes and the success probability of the $\delta$-estimate direction protocol.

\begin{lemma} \label{lem:acaprob}
If a two-node direction estimation protocol is used that transmits $m$ qubits to $\delta$ approximate a direction which succeeds with probability  $\qsucc\geq (1-e^{-\Omega(m\delta )})$ then  with probability at least $\qsucc^{n+2n^2} \geq1-e^{-\Omega(m\delta^2-\log n)}$, all the direction transmissions of init, echo, [\readyone]~and [\readytwo]~type messages are successful. 
\end{lemma}

\begin{proof}
There are at most $n$ init messages, $n^2$ echo messages and $n^2$ [\readyone]~or [\readytwo]~type messages exchanged in the protocol. With probability at least $\qsucc^{n+2n^2}$  all of these transmissions which are between correct nodes are successful.
Now, 
\begin{align}
\qsucc^{n+2n^2} 	&\geq (1-e^{-\Omega(m\delta^2)})^{n+2n^2},\\
				\label{eq:lemprobBern}&\geq 1-(n+2n^2)e^{-\Omega(m\delta^2)},\\
				&\geq 1-e^{-\Omega(m\delta^2-\log n)}
\end{align}

Here inequality~\eqref{eq:lemprobBern} follows using Bernoulli's inequality, which is, $(1+x)^r \geq 1+rx$ for all real $x\geq -1$ and integer $r\geq 2$.

\end{proof}

We see that, theorem~\ref{thm:arc} follows from lemma~\ref{lem:acastterm1}, \ref{lem:acastcon}, \ref{lem:acastterm2}, ~\ref{lem:arcmeaning} and \ref{lem:acaprob}.

%% file: aagree.tex
\subsection {Asynchronous Agreement}

So far we have presented an asynchronous broadcast protocol where a designated sender initiates the protocol with a direction. One major weakness of the protocol is that, if the sender is faulty then the protocol might never terminate, because in this case the correct nodes cannot decide whether the sender is faulty and not sending the [init] message, or correct but very slow. On the other hand, in an asynchronous reference frame agreement protocol the main goal is to allow the correct nodes to agree on some direction despite the presence of---up to a certain number of---unidentified faulty nodes in the network. This requires extra caution to make sure that the protocol eventually terminates. We show that our protocol~\ref{async:agreement} \pname{A-Agree} successfully solves this problem by proving theorem~\ref{thm:aagree}. We repeat the theorem here. 

\setcounter{theorem}{0}
\begin{theorem} 
In a complete network of $n$ nodes that are pairwise connected by public authenticated classical and quantum channels, if a bipartite $\delta$-estimate direction protocol that uses $m$ qubits to achieve success probability $\qsucc \geq 1-e^{-\Omega(m{\delta^2})}$ is used, then protocol \pname{A-Agree} is a $42 \delta$-asynchronous reference frame agreement protocol with success probability at least $1-e^{-\Omega(m{\delta^2}-\log n)}$, that can tolerate up to $t<n/4$ faulty nodes.
 \end{theorem}

There are three epochs in protocol~\ref{async:agreement}. Any correct node that successfully terminates must start at Epoch 0 and terminate at Epoch 3. At each Epoch the nodes inside it, and all the messages transmitted and received by the node while in that Epoch satisfies some invariance properties. We describe and prove these properties in the following lemmas. We first show that a correct node will eventually enter Epoch 1. 

\begin{lemma} \label{lem:AAcomEP0}
For $t<n/4$, all the correct nodes eventually enter Epoch~1 of \pname{A-Agreement} with probability at least  $\qsucc^{n^2+2n^3} \geq1-e^{-\Omega(m\delta^2-\log n)}$.
\end{lemma}
\begin{proof}
%There are $n$ runs of \pname{AR-Cast}s in this step. 
Each of the $n$ nodes has initiated an \pname{AR-Cast} in Epoch 0.  Each of the \pname{AR-Cast}s has a success probability at least $\qsucc^{n+2n^2}$.
So, with probability at least $\qsucc^{n^2+2n^3}$ all the \pname{AR-Cast}s from correct senders are successful. From Lemma~\ref{lem:acaprob} this is at least $1-e^{-\Omega(m\delta^2-\log n)}$.

As $t<n/4$, there are at least $(3t+1)$ correct nodes who initiates \pname{AR-Cast} as sender. According to Theorem~\ref{thm:arc} these $(3t+1)$ \pname{AR-Cast}s will eventually terminate. So, every correct receiver will eventually receive at least $(3t+1)$ directions and go to Epoch~1 with probability at least $\qsucc^{n^2+2n^3}$. 
\end{proof}

Each of the correct nodes stores the output of the \pname{Asynchronous-IC} protocol in an array $b_i$. Here $b_i$ can be seen as an $n\times n$ matrix of bits where row $j$ is received from node $j$. We can observe the following property of this matrix. 

\begin{lemma} \label{aexistt}
For $t<n/4$ and correct node $P_i$, after instruction~\ref{AAE19} of Epoch~1 of \pname{A-Agreement}, there exists a column in $b_i$ with at least $(t+1)$ $1$s in it.
\end{lemma}

\begin{proof}
We show this by a counting argument. Note that a correct node arrives at Epoch~1 only after it have received at least $(3t+1)$ directions from other players. As a result after step~\ref{AAE1a} of Epoch~1 $a_i$ contains at least $(3t+1)$ 1's. These $a_i$'s become the rows of $b_i$ after step~\ref{AAE19}. There are at most $t$ faulty nodes. So, at least $(3t+1)$ rows of $b_i$ are originated from correct nodes. Each of these rows must contain at least $(3t+1)$ $1$'s. So $b_i$ has at least $(3t+1)^2$ $1$s.

However, if no column had at least $(t+1)$ $1$s, then there would be at most $(4t+1)\times t$  $1$s in $b_i$. This contradicts the fact that $b_i$ has at least $(3t+1)^2$ $1$s. So, there must exist a column with at least $(t+1)$ $1$s in it.
\end{proof}

We show that all the correct nodes select the same column which has at least $t+1$ $1$s in it. 

\begin{lemma}
After instruction~\ref{AAE2scan} of Epoch~2 of \pname{A-Agreement}, if correct node $P_i$ has $k_i$ and correct node $P_j$ has $k_j$, then $k_i = k_j$.
\end{lemma}

\begin{proof}
After completion of protocol \pname{Asynchronous-IC} in Epoch~1, all the correct nodes compute the same output vector. That is,  $b_i = b_j$ for all correct $P_i$ and $P_j$. Also, from lemma~\ref{aexistt} we know there exists a column in $b_i$ with at least $(t+1)$ 1s. So, in Epoch~2 step~\ref{AAE2scan} when correct node $P_i$ and $P_j$ selects $k_i$ and $k_j$ to be the chronologically smallest column index that has at least $(t+1)$ $1$s. They select the same column. i.e.,  $k_i = k_j$.
\end{proof}

Now that every correct node $P_i$ agrees on a column $k_i$ of $b_i$, we observe that.

\begin{lemma} \label{lem:aanowait}
If a correct node $P_i$ selects $k_i$ in instruction~\ref{AAE2scan} of Epoch~2, then the \pname{AR-Cast} initiated by $P_{k_i}$ in Epoch~0 eventually completes successfully.
\end{lemma}

\begin{proof}
We show this by showing that at least one correct node has completed the \pname{AR-Cast} initiated by $P_{k_i}$. Then the lemma follows from the termination condition of \pname{AR-Cast}.

Each row $b_i[j]$ represents $P_i$'s knowledge of which \pname{AR-Cast}s are successfully received by $P_j$. For example, if $b_i[j][l] = 1$, then it means node $P_j$ has reported to $P_i$ that it has completed the \pname{AR-Cast} initiated by node $P_l$ in Epoch~0. If there are at least $(t+1)$ $1$s in the $k_i$th column of $b_i$, it means that there are $(t+1)$ nodes who report that they have received the \pname{AR-Cast} initiated by node $P_{k_i}$ in Epoch~0. At least one of these reports is from a correct node. So, from the termination condition of \pname{AR-Cast} (Lemma~\ref{lem:acastcon}) all the correct nodes eventually successfully complete the \pname{AR-Cast} by $P_k$.  
\end{proof}

Now we are ready to prove {\bf theorem~\ref{thm:aagree}.}
\begin{proof}\
There are at most $n$ \pname{AR-Cast}s initiated in Epoch~0 of which $(n-t)$ are by correct nodes. From lemma~\ref{lem:acaprob} each of these succeeds with probability $\qsucc^{n+2n^2} \geq 1-e^{-\Omega(m\delta^2-\log n)}$. So all the correct \pname{AR-Casts} succeed with, 
\begin{align}
\qsucc^{n^2+2n^3} &\geq(1-e^{-\Omega(m\delta^2-\log n)})^n,\\
&\geq 1-e^{-\Omega(m\delta^2-\log n)}. \label{eq:aabern2}
\end{align}
Here inequality~\eqref{eq:aabern2} follows from Bernoulli's inequality. Conditioned on this we show,
\paragraph{Correctness.} 
To prove consistency we show that if a correct node $P_i$ outputs $v_i$ and a correct node $P_j$ outputs $v_j$ then $d(v_i,v_j) \leq 42\delta$. 
From step~\ref{AAE2asg} of Epoch~2 of \pname{A-Agree} we see that, 

\begin{align}
\label{eq:vwcasti}v_i = w_i[k_i],\\ 
\label{eq:vwcastj}v_j = w_j[k_j].
\end{align}
From lemma~\ref{lem:acastcon} we know that for $t<n/4$, 
\begin{align}
d(w_i[k_i],w_j[k_j]) \leq 42\delta. 
\end{align}
This with~\eqref{eq:vwcasti} and~\eqref{eq:vwcastj} gives, 
\begin{align}
d(v_i,v_j) \leq 42\delta. 
\end{align}

\paragraph{Termination} To prove termination we have to show that every correct node $P_i$ terminates with an output direction $v_i$.

To prove this we show that $P_i$ eventually completes all the Epochs of \pname{A-Agree}. From Lemma~\ref{lem:AAcomEP0} we see that $P_i$ must enter Epoch~1 from Epoch~0. All the steps in Epoch~1 are of constant expected time. So, a correct node will eventually complete them and go to Epoch~2. Only in step~\ref{AAE2wait} of Epoch~2 $P_i$ waits for completion of \pname{AR-Cast} from $P_{k_i}$. However, from Lemma~\ref{lem:aanowait}
we know that this \pname{AR-Cast} eventually successfully completes. All the other incomplete \pname{AR-Cast}s are then aborted at Step~\ref{AAE2abort} and the protocol terminates with output $v_i$.

\end{proof}

%% file: asynRefArxiv3.bbl
\begin{thebibliography}{31}
\expandafter\ifx\csname natexlab\endcsname\relax\def\natexlab#1{#1}\fi
\expandafter\ifx\csname bibnamefont\endcsname\relax
  \def\bibnamefont#1{#1}\fi
\expandafter\ifx\csname bibfnamefont\endcsname\relax
  \def\bibfnamefont#1{#1}\fi
\expandafter\ifx\csname citenamefont\endcsname\relax
  \def\citenamefont#1{#1}\fi
\expandafter\ifx\csname url\endcsname\relax
  \def\url#1{\texttt{#1}}\fi
\expandafter\ifx\csname urlprefix\endcsname\relax\def\urlprefix{URL }\fi
\providecommand{\bibinfo}[2]{#2}
\providecommand{\eprint}[2][]{\url{#2}}

\bibitem[{\citenamefont{Kimble}(2008)}]{kim08}
\bibinfo{author}{\bibfnamefont{H.~J.} \bibnamefont{Kimble}},
  \bibinfo{journal}{Nature} \textbf{\bibinfo{volume}{453}},
  \bibinfo{pages}{1023} (\bibinfo{year}{2008}).

\bibitem[{\citenamefont{Aspelmeyer et~al.}(2003)\citenamefont{Aspelmeyer,
  Jennewein, Pfennigbauer, Leeb, and Zeilinger}}]{AJP+03}
\bibinfo{author}{\bibfnamefont{M.}~\bibnamefont{Aspelmeyer}},
  \bibinfo{author}{\bibfnamefont{T.}~\bibnamefont{Jennewein}},
  \bibinfo{author}{\bibfnamefont{M.}~\bibnamefont{Pfennigbauer}},
  \bibinfo{author}{\bibfnamefont{W.}~\bibnamefont{Leeb}}, \bibnamefont{and}
  \bibinfo{author}{\bibfnamefont{A.}~\bibnamefont{Zeilinger}},
  \bibinfo{journal}{IEEE J. Sel. Topics Quantum Electron.}
  \textbf{\bibinfo{volume}{9}}, \bibinfo{pages}{1541} (\bibinfo{year}{2003}).

\bibitem[{\citenamefont{Bonato et~al.}(2009)\citenamefont{Bonato, Tomaello,
  Deppo, Naletto, and Villoresi}}]{BTD+09}
\bibinfo{author}{\bibfnamefont{C.}~\bibnamefont{Bonato}},
  \bibinfo{author}{\bibfnamefont{A.}~\bibnamefont{Tomaello}},
  \bibinfo{author}{\bibfnamefont{V.~D.} \bibnamefont{Deppo}},
  \bibinfo{author}{\bibfnamefont{G.}~\bibnamefont{Naletto}}, \bibnamefont{and}
  \bibinfo{author}{\bibfnamefont{P.}~\bibnamefont{Villoresi}},
  \bibinfo{journal}{New J. Phys.} \textbf{\bibinfo{volume}{11}},
  \bibinfo{pages}{045017} (\bibinfo{year}{2009}).

\bibitem[{\citenamefont{Peng et~al.}(2005)\citenamefont{Peng, Yang, Bao, Zhang,
  Jin, Feng, Yang, Yang, Yin, Zhang et~al.}}]{PYB+05}
\bibinfo{author}{\bibfnamefont{C.-Z.} \bibnamefont{Peng}},
  \bibinfo{author}{\bibfnamefont{T.}~\bibnamefont{Yang}},
  \bibinfo{author}{\bibfnamefont{X.-H.} \bibnamefont{Bao}},
  \bibinfo{author}{\bibfnamefont{J.}~\bibnamefont{Zhang}},
  \bibinfo{author}{\bibfnamefont{X.-M.} \bibnamefont{Jin}},
  \bibinfo{author}{\bibfnamefont{F.-Y.} \bibnamefont{Feng}},
  \bibinfo{author}{\bibfnamefont{B.}~\bibnamefont{Yang}},
  \bibinfo{author}{\bibfnamefont{J.}~\bibnamefont{Yang}},
  \bibinfo{author}{\bibfnamefont{J.}~\bibnamefont{Yin}},
  \bibinfo{author}{\bibfnamefont{Q.}~\bibnamefont{Zhang}},
  \bibnamefont{et~al.}, \bibinfo{journal}{Phys. Rev. Lett.}
  \textbf{\bibinfo{volume}{94}}, \bibinfo{pages}{150501}
  (\bibinfo{year}{2005}).

\bibitem[{\citenamefont{Bonato et~al.}(2006)\citenamefont{Bonato, Aspelmeyer,
  Jennewein, Pernechele, Villoresi, and Zeilinger}}]{BAM+06}
\bibinfo{author}{\bibfnamefont{C.}~\bibnamefont{Bonato}},
  \bibinfo{author}{\bibfnamefont{M.}~\bibnamefont{Aspelmeyer}},
  \bibinfo{author}{\bibfnamefont{T.}~\bibnamefont{Jennewein}},
  \bibinfo{author}{\bibfnamefont{C.}~\bibnamefont{Pernechele}},
  \bibinfo{author}{\bibfnamefont{P.}~\bibnamefont{Villoresi}},
  \bibnamefont{and}
  \bibinfo{author}{\bibfnamefont{A.}~\bibnamefont{Zeilinger}},
  \bibinfo{journal}{Opt. Express} \textbf{\bibinfo{volume}{14}},
  \bibinfo{pages}{10050} (\bibinfo{year}{2006}).

\bibitem[{\citenamefont{Armengol et~al.}(2008)\citenamefont{Armengol, Furch,
  de~Matos, Minster, Cacciapuoti, Pfennigbauer, Aspelmeyer, Jennewein, Ursin,
  Schmitt-Manderbach et~al.}}]{AFJ+08}
\bibinfo{author}{\bibfnamefont{J.~M.~P.} \bibnamefont{Armengol}},
  \bibinfo{author}{\bibfnamefont{B.}~\bibnamefont{Furch}},
  \bibinfo{author}{\bibfnamefont{C.~J.} \bibnamefont{de~Matos}},
  \bibinfo{author}{\bibfnamefont{O.}~\bibnamefont{Minster}},
  \bibinfo{author}{\bibfnamefont{L.}~\bibnamefont{Cacciapuoti}},
  \bibinfo{author}{\bibfnamefont{M.}~\bibnamefont{Pfennigbauer}},
  \bibinfo{author}{\bibfnamefont{M.}~\bibnamefont{Aspelmeyer}},
  \bibinfo{author}{\bibfnamefont{T.}~\bibnamefont{Jennewein}},
  \bibinfo{author}{\bibfnamefont{R.}~\bibnamefont{Ursin}},
  \bibinfo{author}{\bibfnamefont{T.}~\bibnamefont{Schmitt-Manderbach}},
  \bibnamefont{et~al.}, \bibinfo{journal}{Acta Astronaut.}
  \textbf{\bibinfo{volume}{63}}, \bibinfo{pages}{165 } (\bibinfo{year}{2008}).

\bibitem[{\citenamefont{Sasaki et~al.}(2011)\citenamefont{Sasaki, Fujiwara,
  Ishizuka, Klaus, Wakui, Takeoka, Miki, Yamashita, Wang, Tanaka
  et~al.}}]{SFI+11}
\bibinfo{author}{\bibfnamefont{M.}~\bibnamefont{Sasaki}},
  \bibinfo{author}{\bibfnamefont{M.}~\bibnamefont{Fujiwara}},
  \bibinfo{author}{\bibfnamefont{H.}~\bibnamefont{Ishizuka}},
  \bibinfo{author}{\bibfnamefont{W.}~\bibnamefont{Klaus}},
  \bibinfo{author}{\bibfnamefont{K.}~\bibnamefont{Wakui}},
  \bibinfo{author}{\bibfnamefont{M.}~\bibnamefont{Takeoka}},
  \bibinfo{author}{\bibfnamefont{S.}~\bibnamefont{Miki}},
  \bibinfo{author}{\bibfnamefont{T.}~\bibnamefont{Yamashita}},
  \bibinfo{author}{\bibfnamefont{Z.}~\bibnamefont{Wang}},
  \bibinfo{author}{\bibfnamefont{A.}~\bibnamefont{Tanaka}},
  \bibnamefont{et~al.}, \bibinfo{journal}{Opt. Express}
  \textbf{\bibinfo{volume}{19}}, \bibinfo{pages}{10387} (\bibinfo{year}{2011}).

\bibitem[{\citenamefont{Cirac et~al.}(1997)\citenamefont{Cirac, Zoller, Kimble,
  and Mabuchi}}]{CZKM}
\bibinfo{author}{\bibfnamefont{J.~I.} \bibnamefont{Cirac}},
  \bibinfo{author}{\bibfnamefont{P.}~\bibnamefont{Zoller}},
  \bibinfo{author}{\bibfnamefont{H.~J.} \bibnamefont{Kimble}},
  \bibnamefont{and} \bibinfo{author}{\bibfnamefont{H.}~\bibnamefont{Mabuchi}},
  \bibinfo{journal}{Phys. Rev. Lett.} \textbf{\bibinfo{volume}{78}},
  \bibinfo{pages}{3221} (\bibinfo{year}{1997}).

\bibitem[{\citenamefont{Elliott}(2002)}]{Elt02}
\bibinfo{author}{\bibfnamefont{C.}~\bibnamefont{Elliott}},
  \bibinfo{journal}{New J. Phys.} \textbf{\bibinfo{volume}{4}},
  \bibinfo{pages}{46} (\bibinfo{year}{2002}).

\bibitem[{\citenamefont{Beals et~al.}(2013)\citenamefont{Beals, Brierley, Gray,
  Harrow, Kutin, Linden, Shepherd, and Stather}}]{BBG+13}
\bibinfo{author}{\bibfnamefont{R.}~\bibnamefont{Beals}},
  \bibinfo{author}{\bibfnamefont{S.}~\bibnamefont{Brierley}},
  \bibinfo{author}{\bibfnamefont{O.}~\bibnamefont{Gray}},
  \bibinfo{author}{\bibfnamefont{A.~W.} \bibnamefont{Harrow}},
  \bibinfo{author}{\bibfnamefont{S.}~\bibnamefont{Kutin}},
  \bibinfo{author}{\bibfnamefont{N.}~\bibnamefont{Linden}},
  \bibinfo{author}{\bibfnamefont{D.}~\bibnamefont{Shepherd}}, \bibnamefont{and}
  \bibinfo{author}{\bibfnamefont{M.}~\bibnamefont{Stather}},
  \bibinfo{journal}{Proc. R. Soc. A} \textbf{\bibinfo{volume}{469}}
  (\bibinfo{year}{2013}).

\bibitem[{\citenamefont{{Li} and {Benjamin}}(2012)}]{LB12}
\bibinfo{author}{\bibfnamefont{Y.}~\bibnamefont{{Li}}} \bibnamefont{and}
  \bibinfo{author}{\bibfnamefont{S.~C.} \bibnamefont{{Benjamin}}},
  \bibinfo{journal}{New Journal of Physics} \textbf{\bibinfo{volume}{14}},
  \bibinfo{eid}{093008} (\bibinfo{year}{2012}), \eprint{1204.0443}.

\bibitem[{\citenamefont{Barz et~al.}(2012)\citenamefont{Barz, Kashefi,
  Broadbent, Fitzsimons, Zeilinger, and Walther}}]{BKB+12}
\bibinfo{author}{\bibfnamefont{S.}~\bibnamefont{Barz}},
  \bibinfo{author}{\bibfnamefont{E.}~\bibnamefont{Kashefi}},
  \bibinfo{author}{\bibfnamefont{A.}~\bibnamefont{Broadbent}},
  \bibinfo{author}{\bibfnamefont{J.~F.} \bibnamefont{Fitzsimons}},
  \bibinfo{author}{\bibfnamefont{A.}~\bibnamefont{Zeilinger}},
  \bibnamefont{and} \bibinfo{author}{\bibfnamefont{P.}~\bibnamefont{Walther}},
  \bibinfo{journal}{Science} \textbf{\bibinfo{volume}{335}},
  \bibinfo{pages}{303} (\bibinfo{year}{2012}).

\bibitem[{\citenamefont{Massar and Popescu}(1995)}]{MP95}
\bibinfo{author}{\bibfnamefont{S.}~\bibnamefont{Massar}} \bibnamefont{and}
  \bibinfo{author}{\bibfnamefont{S.}~\bibnamefont{Popescu}},
  \bibinfo{journal}{Phys. Rev. Lett.} \textbf{\bibinfo{volume}{74}},
  \bibinfo{pages}{1259} (\bibinfo{year}{1995}).

\bibitem[{\citenamefont{Peres and Scudo}(2001{\natexlab{a}})}]{PS01}
\bibinfo{author}{\bibfnamefont{A.}~\bibnamefont{Peres}} \bibnamefont{and}
  \bibinfo{author}{\bibfnamefont{P.~F.} \bibnamefont{Scudo}},
  \bibinfo{journal}{Phys. Rev. Lett.} \textbf{\bibinfo{volume}{87}},
  \bibinfo{pages}{167901} (\bibinfo{year}{2001}{\natexlab{a}}).

\bibitem[{\citenamefont{Bagan et~al.}(2004)\citenamefont{Bagan, Baig,
  Mu{\~n}oz-Tapia, and Rodriguez}}]{BBM04}
\bibinfo{author}{\bibfnamefont{E.}~\bibnamefont{Bagan}},
  \bibinfo{author}{\bibfnamefont{M.}~\bibnamefont{Baig}},
  \bibinfo{author}{\bibfnamefont{R.}~\bibnamefont{Mu{\~n}oz-Tapia}},
  \bibnamefont{and}
  \bibinfo{author}{\bibfnamefont{A.}~\bibnamefont{Rodriguez}},
  \bibinfo{journal}{Phys. Rev. A} \textbf{\bibinfo{volume}{69}},
  \bibinfo{pages}{010304} (\bibinfo{year}{2004}).

\bibitem[{\citenamefont{Chiribella and D'Ariano}(2004)}]{CD04}
\bibinfo{author}{\bibfnamefont{G.}~\bibnamefont{Chiribella}} \bibnamefont{and}
  \bibinfo{author}{\bibfnamefont{G.~M.} \bibnamefont{D'Ariano}},
  \bibinfo{journal}{J. Math. Phys.} \textbf{\bibinfo{volume}{45}},
  \bibinfo{pages}{4435} (\bibinfo{year}{2004}).

\bibitem[{\citenamefont{Bagan and Mu{\~n}oz-Tapia}(2006)}]{BM06}
\bibinfo{author}{\bibfnamefont{E.}~\bibnamefont{Bagan}} \bibnamefont{and}
  \bibinfo{author}{\bibfnamefont{R.}~\bibnamefont{Mu{\~n}oz-Tapia}},
  \bibinfo{journal}{Int. J. Quantum Inf.} \textbf{\bibinfo{volume}{4}},
  \bibinfo{pages}{5} (\bibinfo{year}{2006}).

\bibitem[{\citenamefont{Giovannetti et~al.}(2006)\citenamefont{Giovannetti,
  Lloyd, and Maccone}}]{GLM06}
\bibinfo{author}{\bibfnamefont{V.}~\bibnamefont{Giovannetti}},
  \bibinfo{author}{\bibfnamefont{S.}~\bibnamefont{Lloyd}}, \bibnamefont{and}
  \bibinfo{author}{\bibfnamefont{L.}~\bibnamefont{Maccone}},
  \bibinfo{journal}{Phys. Rev. Lett.} \textbf{\bibinfo{volume}{96}},
  \bibinfo{pages}{010401} (\bibinfo{year}{2006}).

\bibitem[{\citenamefont{Skotiniotis and Gour}(2012)}]{SG12}
\bibinfo{author}{\bibfnamefont{M.}~\bibnamefont{Skotiniotis}} \bibnamefont{and}
  \bibinfo{author}{\bibfnamefont{G.}~\bibnamefont{Gour}}, \bibinfo{journal}{New
  J. Phys.} \textbf{\bibinfo{volume}{14}}, \bibinfo{pages}{073022}
  (\bibinfo{year}{2012}).

\bibitem[{\citenamefont{Islam et~al.}(2014)\citenamefont{Islam, Magnin, Sorg,
  and Wehner}}]{IMSW14}
\bibinfo{author}{\bibfnamefont{T.}~\bibnamefont{Islam}},
  \bibinfo{author}{\bibfnamefont{L.}~\bibnamefont{Magnin}},
  \bibinfo{author}{\bibfnamefont{B.}~\bibnamefont{Sorg}}, \bibnamefont{and}
  \bibinfo{author}{\bibfnamefont{S.}~\bibnamefont{Wehner}},
  \bibinfo{journal}{New J. Phys.} \textbf{\bibinfo{volume}{16}},
  \bibinfo{pages}{063040} (\bibinfo{year}{2014}).

\bibitem[{\citenamefont{Bartlett et~al.}(2007)\citenamefont{Bartlett, Rudolph,
  and Spekkens}}]{BRS07}
\bibinfo{author}{\bibfnamefont{S.~D.} \bibnamefont{Bartlett}},
  \bibinfo{author}{\bibfnamefont{T.}~\bibnamefont{Rudolph}}, \bibnamefont{and}
  \bibinfo{author}{\bibfnamefont{R.~W.} \bibnamefont{Spekkens}},
  \bibinfo{journal}{Rev. Mod. Phys.} \textbf{\bibinfo{volume}{79}},
  \bibinfo{pages}{555} (\bibinfo{year}{2007}).

\bibitem[{\citenamefont{Shadbolt et~al.}(2012)\citenamefont{Shadbolt,
  V{\'e}rtesi, Liang, Branciard, Brunner, and O'Brien}}]{PVL+12}
\bibinfo{author}{\bibfnamefont{P.}~\bibnamefont{Shadbolt}},
  \bibinfo{author}{\bibfnamefont{T.}~\bibnamefont{V{\'e}rtesi}},
  \bibinfo{author}{\bibfnamefont{Y.-C.} \bibnamefont{Liang}},
  \bibinfo{author}{\bibfnamefont{C.}~\bibnamefont{Branciard}},
  \bibinfo{author}{\bibfnamefont{N.}~\bibnamefont{Brunner}}, \bibnamefont{and}
  \bibinfo{author}{\bibfnamefont{J.~L.} \bibnamefont{O'Brien}},
  \bibinfo{journal}{Scientific reports} \textbf{\bibinfo{volume}{2}}
  (\bibinfo{year}{2012}).

\bibitem[{\citenamefont{Brask et~al.}(2013)\citenamefont{Brask, Chaves, and
  Brunner}}]{BCB13}
\bibinfo{author}{\bibfnamefont{J.~B.} \bibnamefont{Brask}},
  \bibinfo{author}{\bibfnamefont{R.}~\bibnamefont{Chaves}}, \bibnamefont{and}
  \bibinfo{author}{\bibfnamefont{N.}~\bibnamefont{Brunner}},
  \bibinfo{journal}{Physical Review A} \textbf{\bibinfo{volume}{88}},
  \bibinfo{pages}{012111} (\bibinfo{year}{2013}).

\bibitem[{\citenamefont{D'Ambrosio et~al.}(2012)\citenamefont{D'Ambrosio,
  Nagali, Walborn, Aolita, Slussarenko, Marrucci, and Sciarrino}}]{DNW+12}
\bibinfo{author}{\bibfnamefont{V.}~\bibnamefont{D'Ambrosio}},
  \bibinfo{author}{\bibfnamefont{E.}~\bibnamefont{Nagali}},
  \bibinfo{author}{\bibfnamefont{S.~P.} \bibnamefont{Walborn}},
  \bibinfo{author}{\bibfnamefont{L.}~\bibnamefont{Aolita}},
  \bibinfo{author}{\bibfnamefont{S.}~\bibnamefont{Slussarenko}},
  \bibinfo{author}{\bibfnamefont{L.}~\bibnamefont{Marrucci}}, \bibnamefont{and}
  \bibinfo{author}{\bibfnamefont{F.}~\bibnamefont{Sciarrino}},
  \bibinfo{journal}{Nature communications} \textbf{\bibinfo{volume}{3}},
  \bibinfo{pages}{961} (\bibinfo{year}{2012}).

\bibitem[{\citenamefont{Komar et~al.}(2014)\citenamefont{Komar, Kessler,
  Bishof, Jiang, Sorensen, Ye, and Lukin}}]{KKB+14}
\bibinfo{author}{\bibfnamefont{P.}~\bibnamefont{Komar}},
  \bibinfo{author}{\bibfnamefont{E.~M.} \bibnamefont{Kessler}},
  \bibinfo{author}{\bibfnamefont{M.}~\bibnamefont{Bishof}},
  \bibinfo{author}{\bibfnamefont{L.}~\bibnamefont{Jiang}},
  \bibinfo{author}{\bibfnamefont{A.~S.} \bibnamefont{Sorensen}},
  \bibinfo{author}{\bibfnamefont{J.}~\bibnamefont{Ye}}, \bibnamefont{and}
  \bibinfo{author}{\bibfnamefont{M.~D.} \bibnamefont{Lukin}},
  \bibinfo{journal}{Nat Phys} \textbf{\bibinfo{volume}{10}},
  \bibinfo{pages}{582} (\bibinfo{year}{2014}).

\bibitem[{\citenamefont{Peres and Scudo}(2001{\natexlab{b}})}]{PS201}
\bibinfo{author}{\bibfnamefont{A.}~\bibnamefont{Peres}} \bibnamefont{and}
  \bibinfo{author}{\bibfnamefont{P.~F.} \bibnamefont{Scudo}},
  \bibinfo{journal}{Phys. Rev. Lett.} \textbf{\bibinfo{volume}{86}},
  \bibinfo{pages}{4160} (\bibinfo{year}{2001}{\natexlab{b}}).

\bibitem[{\citenamefont{Lamport et~al.}(1982)\citenamefont{Lamport, Shostak,
  and Pease}}]{LSP82}
\bibinfo{author}{\bibfnamefont{L.}~\bibnamefont{Lamport}},
  \bibinfo{author}{\bibfnamefont{R.}~\bibnamefont{Shostak}}, \bibnamefont{and}
  \bibinfo{author}{\bibfnamefont{M.}~\bibnamefont{Pease}},
  \bibinfo{journal}{ACM T. Prog. Lang. Sys.} \textbf{\bibinfo{volume}{4}},
  \bibinfo{pages}{382} (\bibinfo{year}{1982}).

\bibitem[{\citenamefont{Pease et~al.}(1980)\citenamefont{Pease, Shostak, and
  Lamport}}]{PSL80}
\bibinfo{author}{\bibfnamefont{M.}~\bibnamefont{Pease}},
  \bibinfo{author}{\bibfnamefont{R.}~\bibnamefont{Shostak}}, \bibnamefont{and}
  \bibinfo{author}{\bibfnamefont{L.}~\bibnamefont{Lamport}},
  \bibinfo{journal}{J. ACM} \textbf{\bibinfo{volume}{27}}, \bibinfo{pages}{228}
  (\bibinfo{year}{1980}).

\bibitem[{\citenamefont{Ben-Or and El-Yaniv}(2003)}]{BE03}
\bibinfo{author}{\bibfnamefont{M.}~\bibnamefont{Ben-Or}} \bibnamefont{and}
  \bibinfo{author}{\bibfnamefont{R.}~\bibnamefont{El-Yaniv}},
  \bibinfo{journal}{DISTRIB COMPUT} \textbf{\bibinfo{volume}{16}},
  \bibinfo{pages}{249} (\bibinfo{year}{2003}).

\bibitem[{\citenamefont{Canetti and Rabin}(1993)}]{CR93}
\bibinfo{author}{\bibfnamefont{R.}~\bibnamefont{Canetti}} \bibnamefont{and}
  \bibinfo{author}{\bibfnamefont{T.}~\bibnamefont{Rabin}}, in
  \emph{\bibinfo{booktitle}{Proc. ACM STOC'93}} (\bibinfo{publisher}{ACM},
  \bibinfo{year}{1993}), pp. \bibinfo{pages}{42--51}, ISBN
  \bibinfo{isbn}{0-89791-591-7}.

\bibitem[{\citenamefont{Chuang}(2000)}]{Chuang06}
\bibinfo{author}{\bibfnamefont{I.~L.} \bibnamefont{Chuang}},
  \bibinfo{journal}{Phys. Rev. Lett.} \textbf{\bibinfo{volume}{85}},
  \bibinfo{pages}{2006} (\bibinfo{year}{2000}).

\end{thebibliography}
